\documentclass{easychair}

\usepackage{doc}

\usepackage{amsthm}
\usepackage{amssymb}
\usepackage{latexsym}
\usepackage{xspace}
\usepackage{algorithm}
\usepackage[noend]{algpseudocode}
\usepackage{mathtools}
\usepackage{pgf}
\usepackage[caption=false]{subfig}
\usepackage{appendix} 
\DeclarePairedDelimiter{\ceil}{\lceil}{\rceil}
\usepackage{tikz}
\usepackage{multirow}
\usepackage{multicol}
\usepackage{hyperref}

\newtheorem{theorem}{Theorem}[section]
\theoremstyle{definition}
\newtheorem{definition}{Definition}[section]

\usepackage{datetime}
\usepackage[pscoord]{eso-pic}

\begin{document}

\title{Parallel Weighted Model Counting with Tensor Networks}
\author{
Jeffrey M. Dudek\and
Moshe Y. Vardi
}
\institute{Rice University\\
  \email{\{jmd11, vardi\}@rice.edu}}

\authorrunning{J. M. Dudek and M. Y. Vardi}
\titlerunning{Parallel Weighted Model Counting with Tensor Networks}

\maketitle


  
\newcommand{\domain}[1]{[#1]}				
\newcommand{\dlabel}[1]{\ell(#1)} 				

\newcommand{\restrict}[1]{\big|_{#1}}		

\newcommand{\tdim}[1]{\mathcal{I}(#1)}		
\newcommand{\rank}[1]{rank(#1)}				
\newcommand{\size}[1]{size(#1)}				

\newcommand{\tntensor}[1]{\mathcal{T}(#1)}	
\newcommand{\tnfree}[1]{\mathcal{F}(#1)}	
\newcommand{\tnbound}[1]{\mathcal{B}(#1)}	
\newcommand{\struct}[1]{\text{struct}(#1)}	

\newcommand{\edge}[1]{\{#1\}}				
\newcommand{\vinc}[2]{\delta_{#1}(#2)}		
\newcommand{\vincf}[1]{\delta_{#1}}	        
\newcommand{\einc}[2]{\epsilon_{#1}(#2)}	
\newcommand{\eincf}[1]{\epsilon_{#1}}	    
\newcommand{\vincs}[2]{\delta_{#1}[#2]}		
\newcommand{\eincs}[2]{\epsilon_{#1}[#2]}	
\newcommand{\E}[1]{\mathcal{E}(#1)}			
\newcommand{\V}[1]{\mathcal{V}(#1)}			
\newcommand{\Lv}[1]{\mathcal{L}(#1)}		
\newcommand{\C}[3]{\mathcal{C}_{#1,#2}(#3)}	
\newcommand{\paritions}[1]{\mathcal{P}(#1)}
\newcommand{\Line}[1]{Line(#1)}

\newcommand{\blocks}[1]{\mathcal{B}(#1)}
\newcommand{\Ind}[0]{\mathbf{Ind}}
\newcommand{\pkg}[1]{\texttt{#1}}
\newcommand{\tool}[1]{\texttt{#1}}

\newcommand{\fv}[0]{z}
\newcommand{\copyt}[0]{\text{COPY}}

\newcommand{\support}[1]{\text{sup}(#1)}
\newcommand{\shortcite}[1]{\cite{#1}} 


\begin{abstract}
A promising new algebraic approach to weighted model counting makes use of tensor networks, following a reduction from weighted model counting to tensor-network contraction. Prior work has focused on analyzing the single-core performance of this approach, and demonstrated that it is an effective addition to the current portfolio of weighted-model-counting algorithms.

In this work, we explore the impact of multi-core and GPU use on tensor-network contraction for weighted model counting. To leverage multiple cores, we implement a parallel portfolio of tree-decomposition solvers to find an order to contract tensors. To leverage a GPU, we use \tool{TensorFlow} to perform the contractions. We compare the resulting weighted model counter on 1914 standard weighted model counting benchmarks and show that it significantly improves the virtual best solver.

\end{abstract}
\section{Introduction}
In \emph{weighted model counting}, the task is to count the total weight, subject to a given weight function, of the set of solutions of input constraints. This fundamental task has applications in probabilistic reasoning, planning, inexact computing, engineering reliability, and statistical physics \cite{Bacchus2003,DH07,GSS08}. The development of model counters that can successfully compute the total weight on large industrial formulas is an area of active research \cite{OD15,Thurley2006}. Although most model counters focus on single-core performance, there have been several parallel model counters, notably the multi-core (unweighted) model counter \tool{countAntom} \cite{BSB15} and the GPU-based weighted model counter \tool{gpuSAT2} \cite{FHWZ18,FHZ19}. 

The parallelization of neural network training and inference has seen massive research across the machine learning and high-performance computing communities \cite{ABCCDDDGII16,JYPPABBBBB17,PGMLJGKLGA19}. Consequently, GPUs give orders of magnitude of speedup over a single core for neural-network algorithms \cite{KSTKPPRS19,NRBHHJN15}. In this work, we aim to directly leverage advances in multi-core and GPU performance designed for neural-network algorithms in the service of weighted model counting. 

\emph{Tensor networks} provide a natural bridge between high-performance computing and weighted model counting. Tensor networks are a tool used across computer science for reasoning about big-data processing, quantum systems, and more \cite{BB17,Cichocki14,Orus14}. A tensor network describes a complex tensor as a computation on many simpler tensors, and the problem of tensor-network \emph{contraction} is to perform this computation. Contraction is a key operation in neural network training and inference \cite{BK07,Hirata03,KKCLA17,VZTGDMVAC18}, and, as such, many of the optimizations for neural networks also apply to tensor-network contraction \cite{KSTKPPRS19,NRBHHJN15,RMGZFZHVL19}.

Moreover, recent work \cite{BMT15,DDV19} has shown that tensor-network contraction can be used to perform weighted model counting using a 3-stage algorithm. First, in the \emph{reduction} stage, a counting instance is reduced to a tensor-network problem. Second, in the \emph{planning} stage, an order to contract tensors in the network is determined. Finally, in the \emph{execution} stage, tensors in the network are contracted according to the best discovered order. The resulting weighted model counter was shown in \cite{DDV19} to be useful as part of a portfolio of existing model counters when evaluated on a single core. In this work, we explore the impact of multiple-core and GPU use on tensor network contraction for weighted model counting. 

The planning stage in \cite{DDV19} was done using a choice of several single-core heuristic tree-decomposition solvers \cite{AMW17,HS18,Tamaki17}. There is little recent work on parallelizing heuristic tree-decomposition solvers. Instead, we implement a parallel portfolio of single-core tree-decomposition solvers and find that this portfolio significantly improves planning on multiple cores. Similar portfolio approaches have been well-studied and shown to be beneficial in the context of SAT solvers \cite{BSS15,XHHL08}. As a theoretical contribution, we prove that branch-decomposition solvers can also be included in this portfolio. Unfortunately, we find that a current branch-decomposition solver does not significantly improve the portfolio.

The execution stage in \cite{DDV19} was done using \tool{numpy} \cite{numpy} and evaluated on a single core. We add an implementation of the execution stage that uses \tool{TensorFlow} \cite{ABCCDDDGII16} to leverage a GPU for large contractions. Since GPU memory is significantly more limited than CPU memory, we add an implementation of \emph{index slicing}. Index slicing is a recent technique from the tensor-network community \cite{CZHNS18,GK20,VBNHRBM19}, analogous to the classic technique of conditioning in Bayesian Network reasoning \cite{darwiche01,dechter99,pearl86,SAS94}, that allows memory to be significantly reduced at the cost of additional time. We find that, while multiple cores do not significantly improve the contraction stage, a GPU provides significant speedup, especially when augmented with index slicing.

We implement our techniques in \tool{TensorOrder2}, a new parallel weighted model counter. We compare \tool{TensorOrder2} to a variety of state-of-the-art counters. We show that the improved \tool{TensorOrder2} is the fastest counter on 11\% of benchmarks after preprocessing \cite{LM14}, outperforming the GPU-based counter \tool{gpuSAT2} \cite{FHZ19}. Thus \tool{TensorOrder2} is useful as part of a portfolio of counters. All code and data are available at  \url{https://github.com/vardigroup/TensorOrder}.

The rest of the paper is organized as follows: in Section~\ref{sec:prelim}, we give background information on weighted model counting, graph decompositions, and tensor networks. In Section~\ref{sec:algorithm}, we discuss the algorithm for performing weighted model counting using tensor networks as outlined by \cite{DDV19}. In Section~\ref{sec:parallel}, we describe parallelization of this algorithm and two related theoretical improvements: planning with branch decompositions, and index slicing. In Section~\ref{sec:experiments}, we implement this parallelization in \tool{TensorOrder2} and analyze its performance experimentally.

\section{Preliminaries}
\label{sec:prelim}

In this section, we give background information on the three concepts we combine in this paper: weighted model counting, graph decompositions, and tensor-network contraction. 

\subsection{Literal-Weighted Model Counting}
\label{sec:wmc}
The task in weighted model counting is to count the total weight, subject to a given weight function, of the set of solutions of input constraints (typically given in CNF). We focus on so-called \emph{literal-weight functions}, where the weight of a solution can be expressed as the product of weights associated with all satisfied literals. Formally:
\begin{definition}[Weighted Model Count]
  Let $\varphi$ be a formula over Boolean variables $X$ and let $W: X \times \{0,1\} \rightarrow \mathbb{R}$ be a function. The \emph{(literal-)weighted model count} of $\varphi$ w.r.t. $W$ is
  $W(\varphi) \equiv \sum_{\tau \in \domain{X}} \varphi(\tau) \cdot \prod_{x \in X} W(x, \tau(x)),$ where $[X]$ is the set of all functions from $X$ to $\{0, 1\}$.
\end{definition}

We focus in this work on weighted model counting, as opposed to \emph{unweighted model counting} where the weight function $W$ is constant. There are a variety of counters \cite{CW16,FHMW17,Thurley2006} that can perform only unweighted model counting and so we do not compare against them. Of particular note here is \tool{countAntom} \cite{BTB15}, a multi-core unweighted model counter. An interesting direction for future work is to explore the integration of weights into \tool{countAntom} and compare with tensor-network-based approaches to weighted model counting.

Existing approaches to weighted model counting can be split broadly into three categories: \emph{direct reasoning}, \emph{knowledge compilation}, and \emph{dynamic programming}. In counters based on direct reasoning (e.g., \tool{cachet} \cite{SBK05}), the idea is to reason directly about the CNF representation of $\varphi$. In counters based on knowledge compilation (e.g. \tool{miniC2D} \cite{OD15} and \tool{d4} \cite{LM17}), the idea is to compile $\varphi$ into an alternative representation on which counting is easy. In counters based on dynamic programming (e.g. \tool{ADDMC} \cite{DPV20} and \tool{gpuSAT2} \cite{FHWZ18,FHZ19}), the idea is to traverse the clause structure of $\varphi$. Tensor-network approaches to counting (e.g. \tool{TensorOrder} \cite{DDV19} and this work) are also based on dynamic programming. Dynamic programming approaches often utilize graph decompositions, which we define in the next section. 


\subsection{Graphs and Graph Decompositions}
A \emph{graph} $G$ has a nonempty set of vertices $\V{G}$, a set of (undirected) edges $\E{G}$, a function $\delta_G: \V{G} \rightarrow 2^{\E{G}}$ that gives the set of edges incident to each vertex, and a function $\epsilon_G: \E{G} \rightarrow 2^{\V{G}}$ that gives the set of vertices incident to each edge. . Each edge must be incident to exactly two vertices, but multiple edges can
exist between two vertices. If $E \subset \E{G}$, let $\eincs{G}{E} = \bigcup_{e \in E} \einc{G}{e}$. 

A \emph{tree} is a simple, connected, and acyclic graph. A \emph{leaf} of a tree $T$ is a vertex of degree one, and we use $\Lv{T}$ to denote the set of leaves of $T$. For every edge $a$ of $T$, deleting $a$ from $T$ yields exactly two trees, whose leaves define a partition of $\Lv{T}$. Let $C_a \subseteq \Lv{T}$ denote an arbitrary element of this partition. A \emph{rooted binary tree} is a tree $T$ where either $|\V{T}| = 1$ or every vertex of $T$ has degree one or three except a single vertex of degree two (called the \emph{root}). If $|\V{T}| > 1$, the \emph{immediate subtrees of $T$} are the two rooted binary trees that are the connected components of $T$ after the root is removed.

In this work, we use three decompositions of a graph as a tree: tree decompositions \cite{RS91}, branch decompositions \cite{RS91}, and carving decompositions \cite{ST94}. All decompose the graph into an \emph{unrooted binary tree}, which is a tree where every vertex has degree one or three. First, we define tree decompositions \cite{RS91}:
\begin{definition} 
	A \emph{tree decomposition} for a graph $G$ is an unrooted binary tree $T$ together with a labeling function $\chi : \V{T} \rightarrow 2^{\V{G}}$ such that: 
	(1) $\bigcup_{n \in \V{T}} \chi(n) = \V{G}$,
	(2) for all $e \in \E{G}$, there exists $n \in \V{T}$ s.t. $\einc{G}{e} \subseteq \chi(n)$, and 
	(3) for all $n, o, p \in \V{T}$, if $p$ is on the path from $n$ to $o$ then $\chi(n) \cap \chi(o) \subseteq \chi(p)$.
	
	
	The \emph{width} of a tree decomposition is $width_t(T, \chi) \equiv \max_{n \in \V{T}} | \chi(n) | - 1.$
\end{definition}
The treewidth of a graph $G$ is the lowest width among all tree decompositions. Next, we define branch decompositions \cite{RS91}:
\begin{definition}
\label{def:branch}
	A \emph{branch decomposition} for a graph $G$ with $\E{G} \neq \emptyset$ is an unrooted binary tree $T$ whose leaves are the edges of $G$, i.e. $\Lv{T} = \E{G}$. 
	
    The \emph{width} of $T$, denoted $width_b(T)$, is the maximum number of vertices in $G$ that are endpoints of edges in both $C_a$ and $\E{G} \setminus C_a$ for all $a \in \E{T}$, i.e.,
	$width_b(T) \equiv \max_{a \in \E{T}} \left| \eincs{G}{C_a} \cap \eincs{G}{\E{G} \setminus C_a} \right|.$
\end{definition}

The branchwidth of a graph $G$ is the lowest width among all branch decompositions. Carving decompositions are the dual of branch decompositions and hence can be defined by swapping the role of $\V{G}$ and $\E{G}$ in Definition \ref{def:branch}.

The treewidth (plus 1) of graph is no smaller than the branchwidth and is bounded from above by $3/2$ times the branchwidth \cite{RS91}. 

Given a CNF formula $\varphi$, a variety of associated graphs have been considered. The \emph{incidence graph} of $\varphi$ is the bipartite graph where both variables and clauses are vertices and edges indicate that the variable appears in the connected clause. The \emph{primal graph} of $\varphi$ is the graph where variables are vertices and edges indicate that two variables appear together in a clause. There are fixed-parameter tractable model counting algorithms with respect to the treewidth of the incidence graph and the primal graph \cite{SS10}. If the treewidth of the primal graph of a formula $\varphi$ is $k$, the treewidth of the incidence graph of $\varphi$ is at most $k+1$ \cite{KV00}.
 


\subsection{Tensors, Tensor Networks, and Tensor-Network Contraction}
\emph{Tensors} are a generalization of vectors and matrices to higher dimensions-- a tensor with $r$ dimensions is a table of values each labeled by $r$ indices. 

Fix a set $\Ind$ and define an \emph{index} to be an element of $\Ind$. For each index $i$ fix a finite set $\domain{i}$ called the \emph{domain} of $i$. An index is analogous to a variable in constraint satisfaction. 
An \emph{assignment} to $I \subseteq \Ind$ is a function $\tau$ that maps each index $i \in I$ to an element of $\domain{i}$. Let $\domain{I}$ denote the set of assignments to $I$, i.e., $\domain{I} = \{\tau: I \rightarrow \bigcup_{i \in I} \domain{i}~\text{s.t.}~\tau(i) \in \domain{i}~\text{for all}~i \in I\}.$


We now formally define tensors as multidimensional arrays of values, indexed by assignments:
\begin{definition}[Tensor] \label{def:tensor}
	A \emph{tensor} $A$ over a finite set of indices (denoted $\tdim{A}$) is a function $A: \domain{\tdim{A}} \rightarrow \mathbb{R}$.
\end{definition}

The \emph{rank} of a tensor $A$ is the cardinality of $\tdim{A}$. The memory to store a tensor (in a dense way) is exponential in the rank. For example, a scalar is a rank 0 tensor, a vector is a rank 1 tensor, and a matrix is a rank 2 tensor. Some other works generalize Definition \ref{def:tensor} by replacing $\mathbb{R}$ with an arbitrary semiring. 

A \emph{tensor network} defines a complex tensor by combining a set of simpler tensors in a principled way. This is analogous to how a database query defines a resulting table in terms of a computation across many tables.

\begin{definition}[Tensor Network]
	\label{def:tensor-contraction-network}
	A \emph{tensor network} $N$ is a nonempty set of tensors across which no index appears more than twice.
\end{definition}
The set of indices of $N$ that appear once (called \emph{free indices}) is denoted by $\tnfree{N}$. The set of indices of $N$ that appear twice (called \emph{bond indices}) is denoted by $\tnbound{N}$. 
%


The problem of \emph{tensor-network contraction}, given an input tensor network $N$, is to compute the \emph{contraction} of $N$ by marginalizing all bond indices:
\begin{definition}[Tensor-Network Contraction]
The \emph{contraction} of a tensor network $N$ is a tensor $\tntensor{N}$ with indices $\tnfree{N}$ (the set of free indices of $N$), i.e. a function $\tntensor{N} : \domain{\tnfree{N}} \rightarrow \mathbb{R}$, that is defined for all $\tau \in \domain{\tnfree{N}}$ by
		\begin{equation}
        \label{eqn:contraction} 
        \tntensor{N}(\tau) \equiv \sum_{\rho \in \domain{\tnbound{N}}} \prod_{A \in N} A((\rho \cup \tau)\restrict{\tdim{A}}).
        \end{equation}
\end{definition}

A tensor network $N'$ is a \emph{partial contraction} of a tensor network $N$ if there is a surjective function $f: N \rightarrow N'$ s.t. for every $A \in N'$ we have $\tntensor{f^{-1}(A)} = A$; that is, if every tensor in $N'$ is the contraction of some tensors of $N$. If $N'$ is a partial contraction of $N$, then $\tntensor{N'} = \tntensor{N}$.

 Let $A$ and $B$ be tensors. Their \emph{contraction} $A \cdot B$ is the contraction of the tensor network $\{A, B\}$. If $\tdim{A} = \tdim{B}$, their \emph{sum} $A+B$ is the tensor with indices $\tdim{A}$ whose entries are given by the sum of the corresponding entries in $A$ and $B$.

A tensor network can also be seen as a variant of a factor graph \cite{KFL01} with the additional restriction that no variable appears more than twice. The contraction of a tensor network corresponds to the marginalization of a factor graph \cite{RS17}, which is a a special case of the sum-of-products problem \cite{BDP09,dechter99} and the FAQ problem \cite{KNR16}. The restriction on variable appearance is heavily exploited in tools for tensor-network contraction and in this work, since it allows tensor contraction to be implemented as matrix multiplication and so leverage significant work in high-performance computing on matrix multiplication on CPUs \cite{LHKK77} and GPUs \cite{FSH04}.
\section{Weighted Model Counting with Tensor Networks}
\label{sec:algorithm}
In this section, we discuss the algorithm for literal-weighted model counting using tensor networks as outlined by \cite{DDV19}. This algorithm is presented as Algorithm \ref{alg:wmc-tn} and has three stages.

First, in the \emph{reduction} stage the input formula $\varphi$ and weight function $W$ is transformed into a tensor network $N$. We discuss in more detail in Section \ref{sec:algorithm:reduction}.

Second, in the \emph{planning} stage a plan for contracting the tensor network $N$ is determined. This plan takes the form of a \emph{contraction tree} \cite{EP14}:
\begin{definition}[Contraction Tree] \label{def:contraction-tree}
	Let $N$ be a tensor network. A \emph{contraction tree} for $N$ is a rooted binary tree $T$ whose leaves are the tensors of $N$. 
\end{definition}
The planning stage is allowed to modify the input tensor network $N$ as long as the new tensor network $M$ contracts to an identical tensor. We discuss various heuristics for contraction trees in Section \ref{sec:algorithm:planning}.
Planning in Algorithm \ref{alg:wmc-tn} is an anytime process: we heuristically generate better contraction trees until one is ``good enough'' to use. The trade-off between planning and executing is governed by a parameter $\alpha \in \mathbb{R}$, which we determine empirically in Section \ref{sec:experiments}. 

\begin{algorithm}[t]
	\caption{Weighted Model Counting with Tensor Networks}\label{alg:wmc-tn}
	\hspace*{\algorithmicindent} \textbf{Input:} A CNF formula $\varphi$, weight function $W$, and performance factor $\alpha \in \mathbb{R}$. \\
	\hspace*{\algorithmicindent} \textbf{Output:} $W(\varphi)$, the weighted model count of $\varphi$ w.r.t. $W$.
	\begin{algorithmic}[1]
	    \State $N \gets \Call{Reduce}{\varphi, W}$
	    \Repeat
	    \State $M, T \gets \Call{Plan}{N}$
	    \Until{$\alpha \cdot \Call{TimeCost}{M, T} < \text{elapsed time in seconds}$}
	    \State \Return $\Call{Execute}{M, T}(\emptyset)$
	\end{algorithmic}
\end{algorithm}

Third, in the \emph{execution} stage the chosen contraction tree is used to contract the tensor network. We discuss this algorithm in more detail in Section \ref{sec:algorithm:execution}.

We assert the correctness of Algorithm \ref{alg:wmc-tn} in the following theorem.
\begin{theorem}
\label{thm:alg-correctness}
Let $\varphi$ be a CNF formula and let $W$ be a weight function. 
    Assume:
    (1) $\Call{Reduce}{\varphi, W}$ returns a tensor network $N$ s.t. $\tntensor{N}(\emptyset) = W(\varphi)$,
    (2) $\Call{Plan}{N}$ returns a tensor network $M$ and a contraction tree $T$ for $M$ s.t. $\tntensor{M} = \tntensor{N}$, and
    (3) $\Call{Execute}{M, T}$ returns $\tntensor{M}$ for all tensor networks $M$ and contraction trees $T$ for $M$.
Then Algorithm \ref{alg:wmc-tn} returns $W(\varphi)$.
\end{theorem}
\begin{proof}
By Assumption 3, Algorithm \ref{alg:wmc-tn} returns $\tntensor{M}(\emptyset)$. By Assumption 2, this is equal to $\tntensor{N}(\emptyset)$, which by Assumption 1 is exactly $W(\varphi)$.
\end{proof}

Organizationally, note that we discuss the execution stage in Section \ref{sec:algorithm:execution} before we discuss the planning stage. This is because we must understand how plans are used before we can evaluate various planning algorithms.

\subsection{The Reduction Stage}
\label{sec:algorithm:reduction}
There is a reduction from weighted model counting to tensor-network contraction, as described by the following theorem of \cite{DDV19}:
\begin{figure}[t]
	\centering
	\begin{tikzpicture}
\begin{scope}[every node/.style={thick,draw}]
    \node (C) at (-1,-1) {$B_{\neg x \lor \neg y}$};
    \node (x) at (-1,0) {$A_x$};
    \node (A) at (-1.5,1) {$B_{w \lor x \lor \neg y}$};
    \node (y) at (0,0) {$A_y$};
    \node (w) at (0,1) {$A_w$};
    \node (D) at (1,-1) {$B_{\neg y \lor \neg z}$};
    \node (z) at (1,0) {$A_z$};
    \node (B) at (1.5,1) {$B_{w \lor y \lor z}$};
\end{scope}

\begin{scope}[every node/.style={fill=white,circle},
              every edge/.style={draw=black,very thick}]
    \path [-] (A) edge (w);
    \path [-] (A) edge (x);
    \path [-] (A) edge (y);
    \path [-] (B) edge (w);
    \path [-] (B) edge (y);
    \path [-] (B) edge (z);
    \path [-] (C) edge (x);
    \path [-] (C) edge (y);
    \path [-] (D) edge (y);
    \path [-] (D) edge (z);
\end{scope}
\end{tikzpicture}
	\caption{\label{fig:wmc-example} The tensor network produced by Theorem \ref{thm:wmc-reduction} on $\varphi = (w \lor x \lor \neg y) \land (w \lor y \lor z) \land (\neg x \lor \neg y) \land (\neg y \lor \neg z)$ has 8 tensors (indicated by vertices) and 10 indices (indicated by edges). The $A_*$ tensors correspond to variables and have entries computed from the weight function.  The $B_*$ tensors correspond to clauses and have entries computed from the clause negations.}
\end{figure}
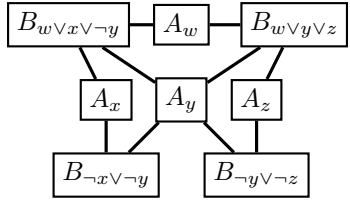
\begin{theorem}
\label{thm:wmc-reduction}
Let $\varphi$ be a CNF formula over Boolean variables $X$ and let $W$ be a weight function. One can construct in polynomial time a tensor network $N_{\varphi,W}$ such that $\tnfree{N_{\varphi,W}} = \emptyset$ and the contraction of $N_{\varphi,W}$ is $W(\varphi)$.
\end{theorem}
\begin{proof}[Sketch of Construction] Let $I = \{(x, C) \in X \times \varphi: x~\text{appears in}~C\}$ be a set of indices, each with domain $\{0,1\}$. The key idea is create a tensor $A_x$ for each variable $x \in X$ and a tensor $B_C$ for each clause $C \in \varphi$ so that each index $(x, C) \in I$ is an index of $A_x$ and $B_C$.

For each $x \in X$, let $A_x$ be a tensor with indices $I \cap (\{x\} \times \varphi)$. For each $\tau \in \domain{I \cap (\{x\} \times \varphi)}$, define $A_x(\tau)$ to be $W(x,0)$ if $\tau$ is always 0, $W(x, 1)$ if $\tau$ is always 1, and 0 otherwise.

For each $C \in \varphi$, let $B_C$ be a tensor with indices $I \cap (X \times \{C\})$. For each $\tau \in \domain{I \cap (X \times \{C\})}$, define $B_C(\tau)$ to be $1$ if $\{x \in X : x~\text{appears in}~C~\text{and}~\tau((x, C)) = 1\}$ satisfies $C$ and 0 otherwise. 



Then $N_{\varphi,W} = \{A_x : x \in X\} \cup \{B_C : C \in \varphi\}$ is the desired tensor network.
\end{proof}

Theorem \ref{thm:wmc-reduction} proves that $\Call{Reduce}{\varphi,W} = N_{\varphi, W}$ satisfies assumption 1 in Theorem \ref{thm:alg-correctness}. See Figure \ref{fig:wmc-example} for an example of the reduction. This reduction is closely related to the formulation of model counting as the marginalization of a factor graph representing the constraints (but assigns tensors to variables as well, not just clauses). 
This reduction can also be extended to other types of constraints, e.g. parity or cardinality constraints.

\subsection{The Execution Stage}
For formulas $\varphi$ with hundreds of clauses, the corresponding tensor network produced by Theorem \ref{thm:wmc-reduction} has hundreds of bond indices. Directly following Equation \ref{eqn:contraction} in this case is infeasible, since it sums over an exponential number of terms. 

Instead, Algorithm \ref{alg:network-contraction} shows how to compute $\tntensor{N}$ for a tensor network $N$ using a contraction tree $T$ as a guide. The key idea is to repeatedly choose two tensors $A_1, A_2 \in N$ (according to the structure of $T$) and contract them. One can prove inductively that Algorithm \ref{alg:network-contraction} satisfies assumption 3 of Theorem \ref{thm:alg-correctness}.


\label{sec:algorithm:execution}
\begin{algorithm}[t]
	\caption{Recursively contracting a tensor network}\label{alg:network-contraction}
	\textbf{Input:} A tensor network $N$ and a contraction tree $T$ for $N$. \\
	\textbf{Output:} $\tntensor{N}$, the contraction of $N$.
	\begin{algorithmic}[1]
	    \Procedure{Execute}{$N,T$}
		\If {$\left|N\right| = 1$}
		\State \Return the tensor contained in $N$
		\Else
        \State $T_1, T_2 \gets \text{immediate subtrees of}~T$
		\State \Return $\Call{Execute}{\Lv{T_1}, T_1} \cdot \Call{Execute}{\Lv{T_2}, T_2}$
		\EndIf
		\EndProcedure
	\end{algorithmic}
\end{algorithm}

Each contraction in Algorithm \ref{alg:network-contraction} contains exactly two tensors and so can be implemented as a matrix multiplication. In \cite{DDV19}, \tool{numpy} \cite{numpy} was used to perform each of these contractions. 
Although the choice of contraction tree does not affect the correctness of Algorithm \ref{alg:network-contraction}, it may have a dramatic impact on the running-time and memory usage. We explore this further in the following section.

\subsection{The Planning Stage}
\label{sec:algorithm:planning}
The task in the planning stage is, given a tensor network, to find a contraction tree that minimizes the computational cost of Algorithm \ref{alg:network-contraction}.

Several \emph{cost-based} approaches \cite{Hirata03,PHV14} aim for minimizing the total number of floating point operations to perform Algorithm \ref{alg:network-contraction}. Most cost-based approaches are designed for tensor networks with very few, large tensors, whereas weighted model counting benchmarks produce tensor networks with many, small tensors. These approaches are thus inappropriate for counting.

Instead, we focus here on \emph{structure-based} approaches, which analyze the rank of intermediate tensors that appear during Algorithm \ref{alg:network-contraction}. The ranks indicate the amount of memory and computation required at each recursive stage. The goal is then to find a contraction tree with small \emph{max-rank}, which is the maximum rank of all tensors that appear during Algorithm \ref{alg:network-contraction}.

This is done through analysis of the \emph{structure graph}, a representation of a tensor network as a graph where tensors are vertices and indices indicate the edges between tensors~\cite{DDV19,MS08,Ying17}:
\begin{definition}[Structure Graph]\label{def:structure}
	Let $N$ be a tensor network. The \emph{structure graph} of $N$ is the graph, denoted $G = \struct{N}$, where $\V{G} = N \cup \{\fv\}$ ($\fv$ is a fresh vertex called the \emph{free vertex}), $\E{G} = \tnbound{N} \cup \tnfree{N}$, $\vinc{G}{\fv} = \tnfree{N}$, and, for all $A \in N$, $\vinc{G}{A} = \tdim{A}$.
	
\end{definition}


For example, if $\varphi$ is a formula and $N_{\varphi,W}$ is the tensor network resulting from Theorem \ref{thm:wmc-reduction}, then $\struct{N_{\varphi,W}}$ is the incidence graph of $\varphi$ (and is independent of the weight function).

One structure-based approach for finding contraction trees \cite{DFGHSW18,MS08} utilizes low-width tree decompositions of the line graph of $\struct{N}$. This approach is analogous to \emph{variable elimination} on factor graphs \cite{BDP09,KDLD05}, which uses tree decompositions of the primal graph of a factor graph.

Unfortunately, for many tensor network all possible contraction trees have large max-rank. \cite{DDV19} observed that this is the case for many standard counting benchmarks. This is because, for each variable $x$, the number of times $x$ appears in $\varphi$ is a lower bound on the max-rank of all contraction trees of $\Call{Reduce}{\varphi, \cdot}$. Thus traditional tensor network planning algorithms (e.g., \cite{DFGHSW18,KCMR18,MS08}) fail on counting benchmarks that have variables which appear more than $30$ times. 

Instead, \cite{DDV19} relaxed the planning stage by allowing changes to the input tensor network $N$. In particular, \cite{DDV19} used tree decompositions of $\struct{N}$ to factor high-rank, highly-structured tensors as a preprocessing step, following prior methods that split large CNF clauses \cite{SS10_2} and high-degree graph vertices \cite{oliveira18,MS11}. A tensor is highly-structured here if it is \emph{tree-factorable}:
\begin{definition} \label{def:tree-factorable}
A tensor $A$ is \emph{tree factorable} if, for every tree $T$ whose leaves are $\tdim{A}$, there is a tensor network $N_A$ and a bijection $g_A: \V{T} \rightarrow N_A$ s.t. $A$ is the contraction of $N_A$ and:
\begin{enumerate}\itemsep0em 
\item $g_A$ is an isomorphism between $T$ and $\struct{N_A}$ with the free vertex removed,
\item for every index $i$ of $A$, $i$ is an index of $g_A(i)$, and
\item for some index $i$ of $A$, every index $j \in \tnbound{N_A}$ satisfies $|\domain{j}| \leq |\domain{i}|$. 
\end{enumerate}
\end{definition}
Informally, a tensor is tree-factorable if it can be replaced by arbitrary trees of low-rank tensors. All tensors produced by Theorem \ref{thm:wmc-reduction} are tree factorable. 
The preprocessing method of \cite{DDV19} is then formalized as follows:
\begin{theorem} \label{thm:factorable-tree}
Let $N$ be a tensor network of tree-factorable tensors such that $|\tnfree{N}| \leq 3$ and $\struct{N}$ has a tree decomposition of width $w \geq 1$. Then we can construct in polynomial time a tensor network $M$ and a contraction tree $T$ for $M$ s.t. $\text{max-rank}(T) \leq \ceil{4(w+1)/3}$ and $N$ is a partial contraction of $M$.
\end{theorem}

This gives us the $\Call{FactorTree}{N}$ planning method, which computes a tree decomposition of $\struct{N}$ and returns the $M$ and $T$ that result from Theorem \ref{thm:factorable-tree}. Because $N$ is a partial contraction of $M$, $\tntensor{N} = \tntensor{M}$. Since all tensors produced by Thereom \ref{thm:wmc-reduction} are tree-factorable, $\Call{FactorTree}{N}$ satisfies assumption 2 in Theorem \ref{thm:alg-correctness}. 

While both $\Call{FactorTree}{N}$ and variable elimination \cite{BDP09,KDLD05} use tree decompositions, they consider different graphs. For example, consider computing the model count of $\psi = \left(\lor_{i=1}^n x_i\right) \land \left(\lor_{i=1}^n \neg x_i\right)$. $\Call{FactorTree}{N}$ uses tree decompositions of the incidence graph of $\psi$, which has treewidth 2. Variable elimination uses tree decompositions of the primal graph of $\psi$, which has treewidth $n-1$. Thus $\Call{FactorTree}{N}$ can exhibit significantly better behavior than variable elimination on some formulas. On the other hand, the behavior of $\Call{FactorTree}{N}$ is at most slightly worse than variable elimination: as noted above, if the treewidth of the primal graph of a formula $\varphi$ is $k$, the treewidth of the incidence graph of $\varphi$ is at most $k+1$ \cite{KV00}. 
\section{Parallelizing Tensor-Network Contraction}
\label{sec:parallel}

In this section, we discuss opportunities for parallelization of Algorithm \ref{alg:wmc-tn}. Since the reduction stage was not a significant source of runtime in \cite{DDV19}, we focus on opportunities in the planning and execution stages.

\subsection{Parallelizing the Planning Stage with a Portfolio}
\label{sec:parallel:planning}
As discussed in Section \ref{sec:algorithm:planning}, $\Call{FactorTree}{N}$ first computes a tree decomposition of $\struct{N}$ and then applies Theorem \ref{thm:factorable-tree}. In \cite{DDV19}, most time in the planning stage was spent finding low-width tree decompositions with a choice of single-core heuristic tree-decomposition solvers \cite{AMW17,HS18,Tamaki17}. Finding low-width tree decompositions is thus a valuable target for parallelization.

We were unable to find state-of-the-art parallel heuristic tree-decomposition solvers. There are several classic algorithms for parallel tree-decomposition construction \cite{Lagergren90,SWG13}, but no recent implementations. There is also an exact tree-decomposition solver that leverages a GPU \cite{VB17}, but exact tools are not able to scale to handle our benchmarks. Existing single-core heuristic tree-decomposition solvers are highly optimized and so parallelizing them is nontrivial.

Instead, we take inspiration from the SAT solving community. One competitive approach to building a parallel SAT solver is to run a variety of SAT solvers in parallel across multiple cores \cite{BSS15,MSSS13,XHHL08}. The portfolio SAT solver \tool{CSHCpar} \cite{MSSS13} won the open parallel track in the 2013 SAT Competition with this technique, and such portfolio solvers are now banned from most tracks due to their effectiveness relative to their ease of implementation. Similar parallel portfolios are thus promising for integration into the planning stage.

We apply this technique to heuristic tree-decomposition solvers by running a variety of single-core heuristic tree-decomposition solvers in parallel on multiple cores and collating their outputs into a single stream. We analyze the performance of this technique in Section \ref{sec:experiments}.

While Theorem \ref{thm:factorable-tree} lets us integrate tree-decomposition solvers into the portfolio, the portfolio might also be improved by integrating solvers of other graph decompositions. The following theorem shows that branch-decomposition solvers may also be used for $\Call{FactorTree}{N}$:
\begin{theorem} \label{thm:factorable-branch}
Let $N$ be a tensor network of tree-factorable tensors such that $|\tnfree{N}| \leq 3$ and $G=\struct{N}$ has a branch decomposition $T$ of width $w \geq 1$. Then we can construct in polynomial time a tensor network $M$ and a contraction tree $S$ for $M$ s.t. $\text{max-rank}(S) \leq \ceil{4w/3}$ and $N$ is a partial contraction of $M$.
\end{theorem}
\begin{proof}[Sketch of Construction]
For simplicity, we sketch here only the case when $\tnfree{N} = \emptyset$, as occurs in counting. For each $A \in N$, $\tdim{A} = \vinc{G}{A}$ is a subset of the leaves of $T$ and so the smallest connected component of $T$ containing $\tdim{A}$ is a dimension tree $T_A$ of $A$. Factor $A$ with $T_A$ using Definition \ref{def:tree-factorable} to get $N_A$ and $g_A$.

We now construct the contraction tree for $M = \cup_{A \in N} N_A$. For each $n \in \V{T}$, let $M_n = \{B~:~B \in N_A, g_A(B) = n\}$. At each leaf $\ell \in \Lv{T}$, attach an arbitrary contraction tree of $M_\ell$. At each non-leaf $n \in \V{T}$, partition $M_n$ into three equally-sized sets and attach an arbitrary contraction tree for each to the three edges incident to $n$. These attachments create a carving decomposition $T'$ from $T$ for $\struct{M}$, of width no larger than $\ceil{4w/3}$. Finally, apply Theorem 3 of \cite{DDV19} to construct a contraction tree $S$ for $M$ from $T'$ s.t. $\text{max-rank}(S) \leq \ceil{4w/3}$.
\end{proof}

The full proof is an extension of the proof of Theorem \ref{thm:factorable-tree} given in \cite{DDV19} and appears in Appendix \ref{sec:appendix:proof}.
Theorem \ref{thm:factorable-branch} subsumes Theorem \ref{thm:factorable-tree}, since given a graph $G$ and a tree decomposition for $G$ of width $w+1$ one can construct in polynomial time a branch decomposition for $G$ of width $w$ \cite{RS91}. Moreover, on many graphs there are branch decompositions whose width is smaller than all tree decompositions. We explore this potential in practice in Section \ref{sec:experiments}.

It was previously known that branch decompositions can also be used in variable elimination \cite{BDP09}, but Theorem \ref{thm:factorable-branch} considers branch decompositions of a different graph. For example, consider again computing the model count of $\psi = \left(\lor_{i=1}^n x_i\right) \land \left(\lor_{i=1}^n \neg x_i\right)$. Theorem \ref{thm:factorable-branch} uses branch decompositions of the incidence graph of $\psi$, which has branchwidth 2. Variable elimination uses branch decompositions of the primal graph of $\psi$, which has branchwidth $n-1$. 

\subsection{Parallelizing the Execution Stage with TensorFlow and Slicing}
\label{sec:parallelizing:execution}
As discussed in Section \ref{sec:algorithm:execution}, each contraction in Algorithm \ref{alg:network-contraction} can be implemented as a matrix multiplication. This was done in \cite{DDV19} using \tool{numpy} \cite{numpy}, and it is straightforward to adjust the implementation to leverage multiple cores with \tool{numpy} and a GPU with \tool{TensorFlow} \cite{ABCCDDDGII16}.

The primary challenge that emerges when using a GPU is dealing with the small onboard memory. For example, the \tool{NVIDIA Tesla-V100} (which we use in Section \ref{sec:experiments}) has just 16GB of onboard memory. This limits the size of tensors that can be easily manipulated. A single contraction of two tensors can be performed across multiple GPU kernel calls \cite{RRBSKH08}, and similar techniques were implemented in \tool{gpusat2} \cite{FHZ19}. These techniques, however, require parts of the input tensors to be copied into and out of GPU memory, which incurs significant slowdown.

Instead, we use \emph{index slicing} \cite{CZHNS18,GK20,VBNHRBM19} of a tensor network $N$. This technique is analogous to \emph{conditioning} on Bayesian networks \cite{darwiche01,dechter99,pearl86,SAS94}. The idea is to represent $\tntensor{N}$ as the sum of contractions of smaller tensor networks. Each smaller network contains tensor slices:

\begin{definition}
Let $A$ be a tensor, $I \subseteq \textbf{Ind}$, and $\eta \in \domain{I}$. Then the \emph{$\eta$-slice} of $A$ is the tensor $A[\eta]$ with indices $\tdim{A} \setminus I$ defined for all $\tau \in \domain{\tdim{A} \setminus I}$ by $A[\eta](\tau) \equiv A((\tau \cup \eta)\restrict{\tdim{A}}).$
\end{definition}

These are called slices because every value in $A$ appears in exactly one tensor in $\{ A[\eta] : \eta \in \domain{I} \}$. We now define and prove the correctness of index slicing: 
\begin{theorem}
Let $N$ be a tensor network and let $I \subseteq \tnbound{N}$. For each $\eta \in \domain{I}$, let $N[\eta] = \{ A[\eta] : A \in N\}$. Then $\tntensor{N} = \sum_{\eta \in \domain{I}} \tntensor{N[\eta]}.$
\end{theorem}
\begin{proof}
Move the summation over $\eta \in \domain{I}$ to the outside of Equation \ref{eqn:contraction}, then apply the definition of tensor slices and recombine terms.
\end{proof}

\begin{algorithm}[t]
	\caption{Sliced contraction of a tensor network}\label{alg:tn-sliced}
	\textbf{Input:} A tensor network $N$, a contraction tree $T$ for $N$, and a memory bound $m$. \\
	\textbf{Output:} $\tntensor{N}$, the contraction of $N$, performed using at most $m$ memory.
	\begin{algorithmic}[1]
	    \State $I \gets \emptyset$
	    \While{$\Call{MemCost}{N, T, I} > m$}
	    \State $I \gets I \cup \{\Call{ChooseSliceIndex}{N, T, I}\}$
	    \EndWhile
	    \State \Return $\sum_{\eta \in [E]} \Call{Execute}{N[\eta], T[\eta]}$
	\end{algorithmic}
\end{algorithm}

By choosing $I$ carefully, computing each $\tntensor{N[\eta]}$ uses less intermediate memory (compared to computing $\tntensor{N}$) while using the same contraction tree. In exchange, the number of floating point operations to compute all $\tntensor{N[\eta]}$ terms increases. 

Choosing $I$ is itself a difficult problem. Our goal is to choose the smallest $I$ so that contracting each network slice $N[\eta]$ can be done in onboard memory. We first consider adapting Bayesian network conditioning heuristics to the context of tensor networks. Two popular conditioning heuristics are (1) \emph{cutset conditioning} \cite{pearl86}, which chooses $I$ so that each network slice is a tree, and (2) \emph{recursive conditioning} \cite{darwiche01}, which chooses $I$ so that each network slice is disconnected\footnote{The full recursive conditioning procedure then recurses on each connected component. While recursive conditioning is an any-space algorithm, the partial caching required for this is difficult to implement on a GPU.}. Both of these heuristics result in a choice of $I$ far larger than our goal requires. 

Instead, in this work as a first step we use a heuristic from \cite{CZHNS18,GK20}: choose $I$ incrementally, greedily minimizing the memory cost of contracting $N[\eta]$ until the memory usage fits in onboard memory. Unlike cutset and recursive conditioning, the resulting networks $N[\eta]$ are typically still highly connected. One direction for future work is to compare other heuristics for choosing $I$ (e.g., see the discussion in Section 10 of \cite{dechter99}).

This gives us Algorithm \ref{alg:tn-sliced}, which performs the execution stage with limited memory at the cost of additional time. $T[\eta]$ is the contraction tree obtained by computing the $\eta$-slice of every tensor in $T$. $\Call{MemCost}{N, T, I}$ computes the memory for one $\Call{Execute}{N[\eta], T[\eta]}$ call. $\Call{ChooseSliceIndex}{N,T,I}$ chooses the next slice index greedily to minimize memory cost.


\section{Implementation and Evaluation}
\label{sec:experiments}

We aim to answer the following research questions:

\begin{enumerate}\itemsep0em 
    \item[(RQ1)] Is the planning stage improved by a parallel portfolio of decomposition tools?
    
    \item[(RQ2)] Is the planning stage improved by adding a branch-decomposition tool?
    
    \item[(RQ3)] When should Algorithm \ref{alg:wmc-tn} transition from the planning stage to the execution stage (i.e., what should be the value of the performance factor $\alpha$)?
    
    \item[(RQ4)] Is the execution stage improved by leveraging multiple cores and a GPU?
    
    \item[(RQ5)] Do parallel tensor network approaches improve a portfolio of state-of-the-art weighted model counters (\tool{cachet}, \tool{miniC2D}, \tool{d4}, \tool{ADDMC}, and \tool{gpuSAT2})?
\end{enumerate}

We implement our changes on top of \tool{TensorOrder} \cite{DDV19} (which implements Algorithm \ref{alg:wmc-tn}) to produce \tool{TensorOrder2}, a new parallel weighted model counter. Implementation details are described in Section \ref{sec:experiments:impl}. All code is available at  \url{https://github.com/vardigroup/TensorOrder}.

We use a standard set\footnote{\url{https://github.com/vardigroup/ADDMC/releases/tag/v1.0.0}} of 1914 weighted model counting benchmarks \cite{DPV20}. Of these, 1091 benchmarks\footnote{\url{https://www.cs.rochester.edu/u/kautz/Cachet/}} are from Bayesian inference problems \cite{SBK05} and 823 benchmarks\footnote{\url{http://www.cril.univ-artois.fr/KC/benchmarks.html}} are unweighted benchmarks from various domains that were artificially weighted by \cite{DPV20}. For weighted model counters that cannot handle real weights larger than 1 (\tool{cachet} and \tool{gpusat2}), we rescale the weights of benchmarks with larger weights. In Experiment 3, we also consider preprocessing these 1914 benchmarks by applying $\pkg{pmc-eq}$ \cite{LM14} (which preserves weighted model count). 
We evaluate the performance of each tool using the PAR-2 score, which is the sum of of the wall-clock times for each completed benchmark, plus twice the timeout for each uncompleted benchmark.

All counters are run in the Docker images (one for each counter) with Docker 19.03.5. All experiments are run on Google Cloud \tool{n1-standard-8} machines with 8 cores (Intel Haswell, 2.3 GHz) and 30 GB RAM. GPU-based counters are provided an \tool{NVIDIA Tesla V100} GPU (16 GB of onboard RAM) using NVIDIA driver 418.67 and CUDA 10.1.243.

\subsection{Implementation Details of \tool{TensorOrder2}}
\label{sec:experiments:impl}
\tool{TensorOrder2} is primarily implemented in Python 3 as a modified version of \tool{TensorOrder}. We replace portions of the Python code with C++ (\tool{g++} v7.4.0) using Cython 0.29.15 for general speedup, especially in $\Call{FactorTree}{\cdot}$.

\paragraph{Planning.} 
\tool{TensorOrder} contains an implementation of the planning stage using a choice of three single-core tree-decomposition solvers: \pkg{Tamaki} \cite{Tamaki17}, \pkg{FlowCutter} \cite{HS18}, and \pkg{htd} \cite{AMW17}. We add to \tool{TensorOrder2} an implementation of Theorem \ref{thm:factorable-branch} and use it to add a branch-decomposition solver \pkg{Hicks} \cite{hicks02}.
We implement a parallel portfolio of graph-decomposition solvers in C++ and give \tool{TensorOrder2} access to two portfolios, each with access to all cores: \pkg{P3} (which combines \pkg{Tamaki},  \pkg{FlowCutter}, and \pkg{htd}) and \pkg{P4} (which includes \pkg{Hicks} as well).

\paragraph{Execution.} 
\tool{TensorOrder} is able to perform the execution stage on a single core and on multiple cores using \pkg{numpy} v1.18.1 and \pkg{OpenBLAS} v0.2.20. We add to \tool{TensorOrder2} the ability to contract tensors on a GPU with \pkg{TensorFlow} v2.1.0 \cite{ABCCDDDGII16}. To avoid GPU kernel calls for small contractions, \tool{TensorOrder2} uses a GPU only for contractions where one of the tensors involved has rank $\geq 20$, and reverts back to using multi-core \pkg{numpy} otherwise. We also add an implementation of Algorithm \ref{alg:tn-sliced}.
Overall, \tool{TensorOrder2} runs the execution stage on three hardware configurations: \pkg{CPU1} (restricted to a single CPU core), \pkg{CPU8} (allowed to use all 8 CPU cores), and \pkg{GPU} (allowed to use all 8 CPU cores and use a GPU).

\begin{figure}
	\centering
	\input{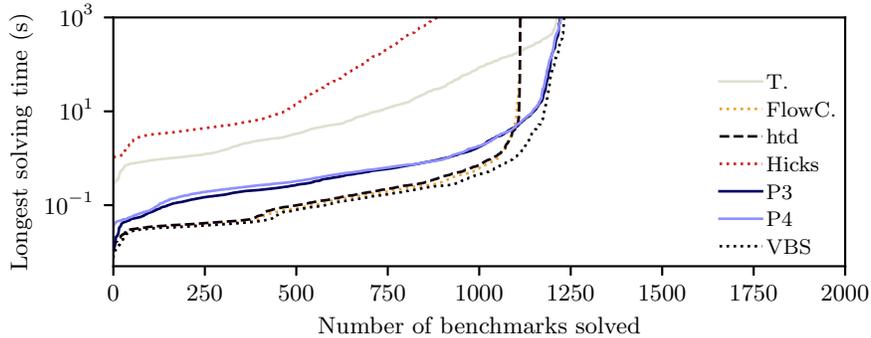}
    \vspace*{-0.7cm}
	\caption{\label{fig:planning} A cactus plot of the performance of various planners. A planner ``solves'' a benchmark when it finds a contraction tree of max rank 30 or smaller.}
\end{figure}

\subsection{Experiment 1: The Planning Stage (RQ1 and RQ2)}
We run each planning implementation (\pkg{FlowCutter}, \pkg{htd}, \pkg{Tamaki}, \pkg{Hicks}, \pkg{P3}, and \pkg{P4}) once on each of our 1914 benchmarks and save all contraction trees found within 1000 seconds (without executing the contractions). Results are summarized in Figure \ref{fig:planning}. 



We observe that the parallel portfolio planners outperform all four single-core planners after 5 seconds. In fact, after 20 seconds both portfolios perform almost as well as the virtual best solver. We conclude that portfolio solvers significantly speed up the planning stage.

We also observe that \pkg{P3} and \pkg{P4} perform almost identically in Figure \ref{fig:planning}. Although after 1000 seconds \pkg{P4} has found better contraction trees than \pkg{P3} on 407 benchmarks, most improvements are small (reducing the max-rank by 1 or 2) or still do not result in good-enough contraction trees. We conclude that adding \pkg{Hicks} improves the portfolio slightly, but not significantly.
 
\subsection{Experiment 2: Determining the Performance Factor (RQ3)}
\label{sec:experiments:pf}
We take each contraction tree discovered in Experiment 1 (with max-rank below 36) and use \tool{TensorOrder2} to execute the tree with a timeout of 1000 seconds on each of three hardware configurations (\pkg{CPU1}, \pkg{CPU8}, and \pkg{GPU}). We observe that the max-rank of almost all solved contraction trees is 30 or smaller.

Given a performance factor, a benchmark, and a planner, we use the planning times from Experiment 1 to determine which contraction tree would have been chosen in step 4 of Algorithm \ref{alg:wmc-tn}. We then add the execution time of the relevant contraction tree on each hardware. In this way, we simulate Algorithm \ref{alg:wmc-tn} for a given planner and hardware with many performance factors. 

\begin{table}[b]
  \caption{\label{tab:performance_factor} The performance factor for each combination of planner and hardware that minimizes the simulated PAR-2 score.}
  \centering
    \begin{tabular}{|l|c|c|c|c|c|c|} \hline
 & \pkg{Tamaki} & \pkg{FlowCutter} & \pkg{htd} & \pkg{Hicks} & \pkg{P3} & \pkg{P4}\\ \hline 
\pkg{CPU1} & $3.8\cdot 10^{-11}$ & $4.8\cdot 10^{-12}$ & $1.6\cdot 10^{-12}$ & $1.0\cdot 10^{-21}$ & $1.4\cdot 10^{-11}$ & $1.6\cdot 10^{-11}$\\ \hline 
\pkg{CPU8} & $7.8\cdot 10^{-12}$ & $1.8\cdot 10^{-12}$ & $1.3\cdot 10^{-12}$ & $1.0\cdot 10^{-21}$ & $5.5\cdot 10^{-12}$ & $6.2\cdot 10^{-12}$\\ \hline 
\pkg{GPU} & $2.1\cdot 10^{-12}$ & $5.5\cdot 10^{-13}$ & $1.3\cdot 10^{-12}$ & $1.0\cdot 10^{-21}$ & $3.0 \cdot 10^{-12}$ & $3.8\cdot 10^{-12}$\\ \hline 
    \end{tabular}
\end{table}

For each planner and hardware, Table 2 shows the performance factor $\alpha$ that minimizes the simulated PAR-2 score. We observe that the performance factor for \pkg{CPU8} is lower than for \pkg{CPU1}, but not necessarily higher or lower than for \pkg{GPU}. We conclude that different combinations of planners and hardware are optimized by different performance factors. 


\subsection{Experiment 3: End-to-End Performance (RQ4 and RQ5)}
Finally, we compare \tool{TensorOrder2} with state-of-the-art weighted model counters \tool{cachet}, \tool{miniC2D}, \tool{d4}, \tool{ADDMC}, and \tool{gpuSAT2}. We consider \tool{TensorOrder2} using \pkg{P4} combined with each hardware configuration (\pkg{CPU1}, \pkg{CPU8}, and \pkg{GPU}), along with \pkg{Tamaki} + \pkg{CPU1} as the best non-parallel configuration from \cite{DDV19}. Note that \tool{P4}+\tool{CPU1} still leverages multiple cores in the planning stage. The performance factor from Experiment 2 is used for each \tool{TensorOrder2} configuration.

We run each counter once on each benchmark (both with and without \pkg{pmc-eq} preprocessing) with a timeout of 1000 seconds and record the wall-clock time taken. When preprocessing is used, both the timeout and the recorded time include preprocessing time. For \tool{TensorOrder2}, recorded times include all of Algorithm \ref{alg:wmc-tn}. Results are summarized in Figure \ref{fig:comparison} and Table \ref{tab:comparison}. 

\begin{figure}[t]
\begin{center}
\input{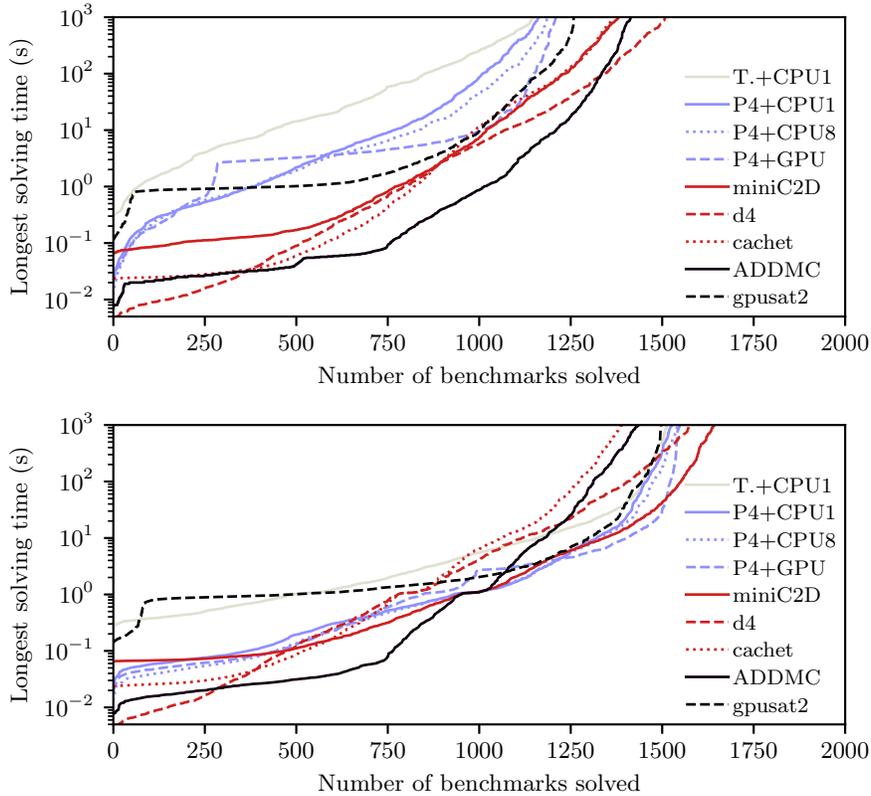}
\vspace*{-0.5cm}
\caption{\label{fig:comparison} A cactus plot of the number of benchmarks solved by various counters, without (above) and with (below) the \tool{pmc-eq} \cite{LM14} preprocessor.}
\end{center}
\vspace*{-0.8cm}
\end{figure}

We observe that the performance of \tool{TensorOrder2} is improved by the portfolio planner and, on hard benchmarks, by executing on a multi-core CPU and on a GPU. The flat line at 3 seconds for $\pkg{P4}+\pkg{GPU}$ is caused by overhead from initializing the GPU.


Comparing \tool{TensorOrder2} with the other counters, \tool{TensorOrder2} is competitive without preprocessing but solves fewer benchmarks than all other counters, although \tool{TensorOrder2} (with some configuration) is faster than all other counters on 158 benchmarks before preprocessing. 
We observe that preprocessing significantly boosts \tool{TensorOrder2} relative to other counters, similar to prior observations with \tool{gpusat2} \cite{FHZ19}. \tool{TensorOrder2} solves the third-most preprocessed benchmarks of any solver and has the second-lowest PAR-2 score (notably, outperforming \tool{gpusat2} in both measures). \tool{TensorOrder2} (with some configuration) is faster than all other counters on 200 benchmarks with preprocessing. Since \tool{TensorOrder2} improves the virtual best solver on 158 benchmarks without preprocessing and on 200 benchmarks with preprocessing, we conclude that \tool{TensorOrder2} is useful as part of a portfolio of counters.

\begin{table}[t]
  \caption{\label{tab:comparison} The numbers of benchmarks solved by each counter fastest and in total after 1000 seconds, and the PAR-2 score.}
  \centering
  \begin{tabular}{l||r|r|r||r|r|r|}
  & \multicolumn{3}{c||}{Without preprocessing} & \multicolumn{3}{c|}{With \tool{pmc-eq} preprocessing} \\
 & \# Fastest & \# Solved & PAR-2 Score & \# Fastest & \# Solved & PAR-2 Score\\ \hline 
\pkg{T.}+\pkg{CPU1} & 0 & 1151 & 1640803. & 0 & 1514 & 834301.\\ 
\pkg{P4}+\pkg{CPU1} & 45 & 1164 & 1562474. & 83 & 1526 & 805547.\\ 
\pkg{P4}+\pkg{CPU8} & 50 & 1185 & 1500968. & 67 & 1542 & 771821.\\ 
\pkg{P4}+\pkg{GPU} & 63 & 1210 & 1436949. & 50 & 1549 & 745659.\\ \hline 
\tool{miniC2D} & 50 & 1381 & 1131457. & 221 & 1643 & 585908.\\ 
\tool{d4} & 615 & 1508 & 883829. & 550 & 1575 & 747318.\\ 
\tool{cachet} & 264 & 1363 & 1156309. & 221 & 1391 & 1099003.\\ 
\tool{ADDMC} & 640 & 1415 & 1032903. & 491 & 1436 & 1008326.\\  
\tool{gpusat2} & 37 & 1258 & 1342646. & 25 & 1497 & 854828.\\ \hline 
\end{tabular}
\end{table}
\section{Discussion}
In this work, we explored the impact of multiple-core and GPU use on tensor network contraction for weighted model counting. We implemented our techniques in \tool{TensorOrder2}, a new parallel counter, and showed that \tool{TensorOrder2} is useful as part of a portfolio of counters.

In the planning stage, we showed that a parallel portfolio of graph-decomposition solvers is significantly faster than single-core approaches. We proved that branch decomposition solvers can also be included in this portfolio, but concluded that a state-of-the-art branch decomposition solver only slightly improves the portfolio. For future work, it would be interesting to consider leveraging other width parameters (e.g. \cite{AGG07} or \cite{GS17}) in the portfolio as well. It may also be possible to improve the portfolio through more advanced algorithm-selection techniques \cite{HHLKS09,XHHL12}. One could develop parallel heuristic decomposition solvers directly, e.g., by adapting the exact GPU-based tree-decomposition solver \cite{VB17} into a heuristic solver.

In the execution stage, we showed that tensor contractions can be performed with \pkg{TensorFlow} on a GPU. When combined with index slicing, we concluded that a GPU speeds up the execution stage for many hard benchmarks. For easier benchmarks, the overhead of a GPU may outweigh any contraction speedups. We focused here on parallelism within a single tensor contraction, but there are opportunities in future work to exploit higher-level parallelism, e.g. by running each slice computation in Algorithm \ref{alg:tn-sliced} on a separate GPU. 

\pkg{TensorFlow} also supports performing tensor contractions on TPUs (tensor processing unit \cite{JYPPABBBBB17}), which are specialized hardware designed for neural network training and inference. Tensor networks therefore provide a natural framework to leverage TPUs for weighted model counting as well. There are additional challenges in the TPU setting: floating-point precision is limited, and there is a (100+ second) compilation stage from a contraction tree into XLA \cite{XLA}. We plan to explore these challenges further in future work.

\section*{Acknowledgments}
We thank Illya Hicks for providing us the code of his branchwidth heuristics. This work was supported in part by the NSF (grants IIS-1527668, CCF-1704883, IIS-1830549, and DMS-1547433), by the DoD (MURI grant N00014-20-1-2787), and by Google Cloud.

\bibliographystyle{plain}
\bibliography{tensornetworks}

\normalsize
\newpage
\appendix
\section{Additional Experimental Results}

\begin{figure}[t]
\begin{center}
\input{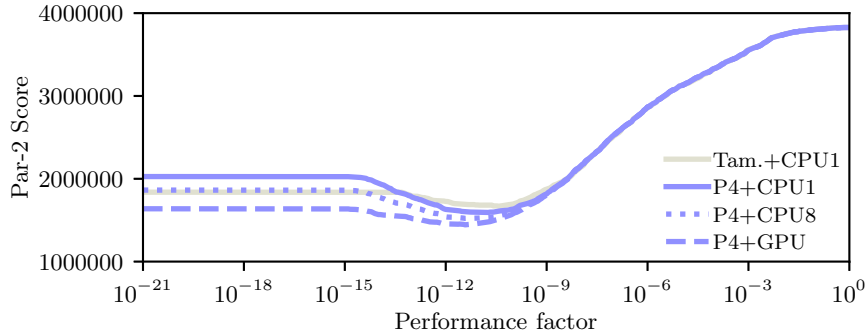}
\caption{\label{fig:performance-factor} A graph of the simulated PAR-2 score for various combinations of planners and hardware as the performance factor varies.}
\end{center}
\end{figure}

Figure \ref{fig:performance-factor} indicates how varying the performance factor affects the simulated PAR-2 score for various combinations of planners and hardware. 

\section{A Proof of Theorem \ref{thm:factorable-branch}}
\label{sec:appendix:proof}

In this section, we present a complete proof of Theorem \ref{thm:factorable-branch}. Note that the proof differs from the proof of Theorem \ref{thm:factorable-tree} in \cite{DDV19} only in Part 1 and Part 4 (and in the definition of $\rho$ in Part 3).

\begin{proof}
The proof proceeds in five steps: (1) compute the factored tensor network $M$, (2) construct a graph $H$ that is a simplified version of the structure graph of $M$, (3) construct a carving decomposition $S$ of $H$, (4) bound the width of $S$, and (5) use $S$ to find a contraction tree for $M$. Working with $H$ instead of directly working with the structure graph of $M$ allows us to cleanly handle tensor networks with free indices.

\textbf{Part 1: Factoring the network.}
Next, for each $v \in \V{G}$, define $T_v$ to be the smallest connected component of $T$ containing $\vinc{G}{v} \subseteq \Lv{T}$. Consider each $A \in N$. If $\tnfree{A} = \emptyset$, let $N_A = \{N_A\}$. Otherwise, observe that $T_A$ is a dimension tree of $A$ and so we can factor $A$ with $T_A$ using Definition \ref{def:tree-factorable} to get a tensor network $N_A$ and a bijection $g_A: \V{T_A} \rightarrow N_A$. Define $M = \cup_{A \in N} N_A$ and let $G'$ be the structure graph of $M$ with free vertex $\fv'$. The remainder of the proof is devoted to bounding the carving width of $G'$. To do this, it is helpful to define $\rho: \V{T} \rightarrow \V{G}$ by $\rho(n) = \{ v \in \V{G} : n \in \V{T_v}, |\vinc{T_v}{n}| = 3\}$. Note that $|\rho(n)| \leq w$ for all $n \in \V{T}$.

\textbf{Part 2: Constructing a simplified structure graph of $M$.} In order to easily characterize $G'$, we define a new, closely-related graph $H$ by taking a copy of $T_v$ for each $v \in \V{G}$ and connecting these copies where indicated by $g$. Formally, the vertices of $H$ are $\{(v, n) : v \in \V{G}, n \in \V{T_v}\}$. For every $v \in \V{G}$ and every arc in $T$ with endpoints $n, m \in \V{T_v}$, we add an edge between $(v, n)$ and $(v, m)$. Moreover, for each $e \in \E{G}$ incident to $v, w \in \V{G}$, we add an edge between $(v, g(e))$ and $(w, g(e))$. 

We will prove in Part 5 that the carving width of $G'$ is bounded from above by the carving width of $H$. We therefore focus in Part 3 and Part 4 on bounding the carving width of $H$. It is helpful for this to define the two projections $\pi_G : \V{H} \rightarrow \V{G}$ and $\pi_T : \V{H} \rightarrow \V{T}$ that indicate respectively the first or second component of a vertex of $H$. 



\textbf{Part 3. Constructing a carving decomposition $S$ of $H$.}
The idea of the construction is, for each $n \in \V{T}$, to attach the elements of $\pi_T^{-1}(n)$ as leaves along the arcs incident to $n$. To do this, for every leaf node $\ell \in \Lv{T}$ with incident arc $a \in \vinc{T}{\ell}$ define $H_{\ell, a} = \pi_T^{-1}(\ell)$. For every non-leaf node $n \in \V{T} \setminus \Lv{T}$ partition $\pi_T^{-1}(n)$ into three sets $\{H_{n,a} : a \in \vinc{T}{n}\}$, ensuring that the degree 3 vertices are divided evenly (the degree 1 and 2 vertices can be placed arbitrarily). Observe that $\{H_{n,a} : n \in \V{T}, a \in \vinc{T}{n}\}$ is a partition of $\V{H}$, and there are at most $\ceil{|\rho(n)|/3}$ vertices of degree 3 in each $H_{n,a}$, since there are exactly $|\rho(n)|$ vertices of degree 3 in $\pi_T^{-1}(n)$. 

We use this to construct a carving decomposition $S$ from $T$ by adding each element of $H_{n,a}$ as a leaf along the arc $a$. Formally, let $x_v$ denote a fresh vertex for each $v \in \V{H}$, let $y_n$ denote a fresh vertex for each $n \in \V{T}$, and let $z_{n,a}$ denote a fresh vertex for each $n \in \V{T}$ and $a \in \vinc{T}{n}$. Define $\V{S}$ to be the union of $\V{H}$ with the set of these free vertices. 

We add an arc between $v$ and $x_v$ for every $v \in \V{H}$. Moreover, for every $a \in \E{T}$ with endpoints $o, p \in \einc{T}{a}$ add an arc between $y_{o,a}$ and $y_{p,a}$. For every $n \in \V{T}$ and incident arc $a \in \vinc{T}{n}$, construct an arbitrary sequence $I_{n,a}$ from $\{x_v : v \in H_{n,a}\}$. If $H_{n,a} = \emptyset$ then add an arc between $y_n$ and $z_{n,a}$. Otherwise, add arcs between $y_n$ and the first element of $I$, between consecutive elements of $I_{n,a}$, and between the last element of $I_{n,a}$ and $z_{n,a}$. 

Finally, remove the previous leaves of $T$ from $S$. The resulting tree $S$ is a carving decomposition of $H$, since we have added all vertices of $H$ as leaves and removed the previous leaves of $T$.

\textbf{Part 4: Computing the width of $S$.} In this part, we separately bound the width of the partition induced by each of the three kinds of arcs in $S$.

First, consider an arc $b$ between some $v \in \V{H}$ and $x_v$. Since all vertices of $H$ are degree 3 or smaller, $b$ defines a partition of width at most $3 \leq \ceil{4w/3}$.

Next, consider an arc $c_a$ between $y_{o,a}$ and $y_{p,a}$ for some arc $a \in \E{T}$ with endpoints $o, p \in \einc{T}{a}$.
Observe that removing $a$ from $T$ defines a partition $\{B_o, B_p\}$ of $\V{T}$, denoted so that $o \in B_o$ and $p \in B_p$. 

Then removing $c_a$ from $S$ defines the partition $\{ \pi_T^{-1}(B_o), \pi_T^{-1}(B_p) \}$ of $\Lv{S}$. By construction of $H$, all edges between $\pi_T^{-1}(B_o)$ and $\pi_T^{-1}(B_p)$ are between $\pi_T^{-1}(o)$ and $\pi_T^{-1}(p)$. Observe that every edge $e \in \E{H}$ between $\pi_T^{-1}(o)$ and $\pi_T^{-1}(p)$ corresponds under $g_v$ to $a$ in $T_v$ for some $v$. It follows that the number of edges between $\pi_T^{-1}(o)$ and $\pi_T^{-1}(o)$ is exactly the number of vertices in $G$ that are endpoints of edges in both $C_a$ and $\E{G} \setminus C_a$, which is bounded by $w$. Thus the partition defined by $c_a$ has width no larger than $w$. 

Finally, consider an arc $d$ added as one of the sequence of $|H_{n,a}|+1$ arcs between $y_n$, $I_{n,a}$, and $z_{n,a}$ for some $n \in \V{T}$ and $a \in \vinc{T}{n}$. Some elements of $H_{n,a}$ have changed blocks from the partition defined by $c_a$. Each vertex of degree 2 that changes blocks does not affect the width of the partition, but each vertex of degree 3 that changes blocks increases the width by 1. There are at most $\ceil{w/3}$ elements of degree 3 added as leaves between $y_n$ and $z_{n,a}$. Thus the partition defined by $d$ has width at most $w + \ceil{w/3} = \ceil{4w/3}$.

It follows that the width of $S$ is at most $\ceil{4w/3}$.

\textbf{Part 5: Bounding the max-rank of $M$.} Let $\fv$ be the free vertex of the structure graph of $N$. We first construct a new graph $H'$ from $H$ by, if $\tnfree{N} \neq \emptyset$, contracting all vertices in $\pi_G^{-1}(\fv)$ to a single vertex $\fv$. If $\tnfree{N} = \emptyset$, instead add $\fv$ as a fresh degree 0 vertex to $H'$. Moreover, for all $A \in N$ with $\tdim{A} = \emptyset$ add $A$ as a degree 0 vertex to $H'$. 

Note that adding degree 0 vertices to a graph does not affect the carving width. Moreover, since $|\tnfree{N}| \leq 3$ all vertices (except at most one) of $\pi_G^{-1}(\fv)$ are degree 2 or smaller. It follows that contracting $\pi_G^{-1}(\fv)$ does not increase the carving width. Thus the carving width of $H'$ is at most $\ceil{4w/3}$.

Moreover, $H'$ and $G'$ are isomorphic. To prove this, define an isomorphism $\phi: \V{H'} \rightarrow \V{G'}$ between $H'$ and $G'$ by, for all $v \in \V{H'}$:
$$\phi(v) \equiv \begin{cases}v&\text{if}~v \in N~\text{and}~\tdim{v}=\emptyset\\\fv'&v=\fv\\g_{\pi_G(v)}(\pi_T(v))&\text{if}~v \in \V{H}~\text{and}~\pi_G(v) \in N\end{cases}$$
$\phi$ is indeed an isomorphism between $H'$ and $G'$ because the functions $g_A$ are all isomorphisms and because an edge exists between $\pi_G^{-1}(v)$ and $\pi_G^{-1}(w)$ for $v, w \in \V{G}$ if and only if there is an edge between $v$ and $w$ in $G$. Thus the carving width of $G'$ is at most $\ceil{4w/3}$. By Theorem 3 of \cite{DDV19}, then, $M$ has a contraction tree of max rank no larger than $\ceil{4w/3}$.
\hfill$\square$
\end{proof}
\end{document}